\documentclass{amsart}
\usepackage{amsfonts}
\usepackage{amsmath}
\usepackage{amssymb}
\usepackage{amsthm}
\usepackage{color}
\usepackage{comment}

\usepackage[foot]{amsaddr}

\usepackage{mathrsfs}

\usepackage{pdfpages}

\usepackage{psfrag}
\usepackage{graphicx}

\newtheorem{Thm}{Theorem}
\newtheorem{Lem}[Thm]{Lemma}

\theoremstyle{definition}
\newtheorem{Def}[Thm]{Definition}

\theoremstyle{remark}
\newtheorem{Rem}[Thm]{Remark}

\numberwithin{equation}{section}

\title[$k$-proportionality]{Between proportionnality and envy-freeness: $k$-proportionality}

\date{\today}

\author[G.~Ch\`eze]{Guillaume Ch\`eze}
\address{Guillaume Ch\`eze: Institut de Math\'ematiques de Toulouse\\
Universit\'e Paul Sabatier \\
118 route de Narbonne\\
31 062 TOULOUSE cedex 9, France}
\email{guillaume.cheze@math.univ-toulouse.fr}

\begin{document}


\begin{abstract}
This article deals with the cake cutting problem. In this setting, there exists two notions of fair division: proportional division (when there are $n$ players, each player thinks to get at least $1/n$ of the cake) and envy-free division (each player wants to keep his or her share because he or she does not envy the portion given to another player). Some results are valid for proportional division but not for envy-free division.\\
Here, we introduce and study a scale between the proportional division and the envy-free division. The goal is to understand where is the gap between statements about proportional division and envy-free division. This scale comes from the notion introduced in this article: $k$-proportionality. When $k=n$ this notion corresponds to the proportional division and when $k=2$ it corresponds to envy-free division.\\
With $k$-proportionality we can understand where some difficulties  in fair division lie. First, we show that there are situations in which there is no $k$-proportional and equitable division of a pie with connected pieces when $k\leq n-1$. This result was known only for envy-free division, ie $k=2$.\\
Next, we prove that there are situations in which there is no Pareto-optimal $k$-proportional division of a cake with connected pieces when $k\leq n-1$. This result was known only for $k=2$.\\
These theorems say that we can get an impossibility result even if we do not consider an envy-free division but a weaker notion.\\
Finally, $k$-proportionality allows to give a generalization with a uniform statement of theorems about strong envy-free and strong proportional divisions. \\

\smallskip
\noindent \textsc{Keywords:} fair division, cake cutting, linear algebra.
\end{abstract}

\maketitle

\section*{Introduction}
In this article we study the fair division of an  heterogeneous good between different players or agents. This good can be for example: a cake, land, time or computer memory. Here, we are going to suppose that we want to share a cake or a pie.\\
Cakes will be modeled by the interval $[0,1]$ and cake portions by the intervals $[x_i,x_{i+1}] \subset [0,1]$.
Pies will be modeled by the circle $S^1$ and we will represent the circle by the interval $[0,1]$ where the  two extremities have been identified.\\

  The problem of fair division has been formulated in a scientific way by Steinhaus in 1948, see \cite{Steinhaus}. Nowadays, there exist several papers, see e.g. \cite{AzizMackenzie, BJK, BramsTaylorarticle, DubinsSpanier, EdmondsPruhs, EvenPaz, Pikhurko}, and books about this topic, see e.g. \cite{Barbanel, BramsTaylor, Procacciachapter, RobertsonWebb}. These results appear in the mathematics, economics, political science, artificial intelligence and computer science literature. Recently, the cake cutting problem has been studied intensively by computer scientists for solving resource allocation problems in multi agents systems, see e.g.~\cite{Chevaleyre06,Sziklai,Dynamic}. \\

 Throughout this article, we will denote by $X$ the cake or the pie that we want to share. We also consider $n$ players and we associate to each player a non-atomic probability measure $\mu_i$ on $X$. Thus  $\mu_i(X)=1$ for all $i$. Moreover, we suppose that the measures $\mu_i$ are  absolutely continuous with respect to the Lebesgue measure.\\
These measures represent the preferences, the utility functions of the players.  This means that if $P$ is a portion of the cake, then $\mu_i(P)$ represents the level of satisfaction obtained by the $i$-th player when he or she gets $P$.\\
 The problem in this situation is to get a fair division of $X=X_1\sqcup \ldots \sqcup X_n$, where the $i$-th player gets $X_i$.\\

When we study fair divisions, we have to define   ``fair" precisely. Indeed, several notions exist.
\begin{enumerate}
\item[$\bullet$] A division is \emph{proportional} when each player gets at least $\mu_i(X)/n$ of the cake for his or her measure. More precisely, we say that a division is proportional when for all $i$, we have 
$$\mu_i(X_i) \geq \mu_i(X)/n.$$
\item[$\bullet$] We say that a division is \emph{strong proportional} when for all $i$, we have 
$$\mu_i(X_i) > \mu_i(X)/n.$$

\item[$\bullet$] A division is \emph{envy-free} when each player thinks his or her share is better than the others'. More precisely, we say that a division is envy-free when for all $i \neq j$, we have 
$$\mu_i(X_i) \geq  \mu_i(X_j).$$
 
\item[$\bullet$] We say that a division is \emph{strong envy-free} when for all $i \neq j$, we have 
$$\mu_i(X_i) >  \mu_i(X_j).$$

\item[$\bullet$] A division is \emph{equitable} when the level of satisfaction obtained by each player is the same. More precisely, we say that a division is equitable when for all $i\neq j$, we have 
$$\mu_i(X_i)=\mu_j(X_j).$$
\item[$\bullet$] A division is \emph{exact} when each player thinks that each portion has a value equal to $1/n.$  More, precisely, we say that a division is exact when for all $i \neq j$, we have 
$$\mu_i(X_j)=\dfrac{1}{n}.$$

\item[$\bullet$] We say that an allocation $A_1$ ($X=\sqcup_{i=1}^n \tilde{X}_i$) \emph{Pareto-dominates} an allocation $A_2$ ($X=\sqcup_{i=1}^n X_i$) if at least one player feels that $A_1$  is better than $A_2$ and no player feels that $A_1$   is worse than $A_2$. More precisely:
$$\exists j : \mu_j(\tilde{X}_j) >\mu_j(X_j) \textrm{ and }\forall i: \mu_i(\tilde{X}_i) \geq  \mu_i(X_i).$$
   \end{enumerate}

 As $\mu_i(X)=1$, we can prove easily that if a division is envy-free then this division is also porportional. Thus, an envy-free division is more demanding than a proportional division. For example, it is possible to get an equitable and proportional division where each $X_j$ is an interval, see \cite{Cechexistence}. However, it is impossible to obtain an equitable and envy-free division with connected pieces of cake, see \cite{BJK2006}.\\

In this article, we introduce and study a scale between the proportional division and the envy-free division. This notion is $k$-proportionality. 

\begin{Def} Let $k$ be an integer such that $2 \leq k \leq n$.
We say that a division \linebreak $X=X_1 \sqcup \ldots \sqcup X_n$ is $k$-proportional when for all $J\subset \{1,\ldots,n\}$ such that $|J|=k$ and all $i \in J$ we have
$$\mu_i(X_i) \geq \dfrac{\sum_{j \in J} \mu_i(X_j)}{k}.$$
\end{Def}

This definition means that a division is $k$-proportional if for all subset \mbox{$J=\{j_1, \ldots,j_k\}$} with $k$ players the division of $X_{j_1} \sqcup \ldots \sqcup X_{j_k}$ is proportional.

\begin{Rem}$\,$
\begin{itemize}
\item[-] An $n$-proportional division is a proportional division since \linebreak \mbox{$\sum_{j=1}^n \mu_i(X_j)=\mu_i(X)=1$}.\\
\item[-] A $2$-proportional division is an envy-free division since
$$\mu_i(X_i) \geq \dfrac{\mu_i(X_i)+\mu_i(X_j)}{2} \iff \mu_i(X_i) \geq \mu_i(X_j).$$
\end{itemize}
\end{Rem}

In the next section, we will see that if  a division is $k$-proportional then it is also $(k+1)$-proportional. \\

As we have mentioned,  some results are valid for proportional division but not for envy-free division, $k$-proportionality allows to understand where the difficulty arises. Indeed, in Section~\ref{sec:equitable},  we are going to prove that it is impossible to get a $k$-proportional and equitable division of a pie with connected pieces when $k\leq n-1$. Thus, we have an impossibility result at every level of the scale except when we consider $k=n$, ie proportional division. Therefore, it is not necessary to consider envy-free divisions to get an impossibility result.
This generalises and improves the result given by J.~Barbanel, S.~Brams and W.~Stromquist in \cite{Barbanel_Brams_Stromquist}.
In Section~\ref{sec:pareto}, we are also going to give the same kind of generalization of an impossibility result about connected envy-free division and Pareto optimality.\\
At last, in Section~\ref{sec:strong}, we give a characterization of strong $k$-proportionality. If $k=2$ or $k=n$, we recover well known results about strong envy-free divisions and strong proportional divisions. This means that with $k$-proportionality we have a general framework to study fair division.\\

\textbf{Related results}:\\
Considering the average $\big(\sum_{j \in J} \mu_i(X_j)\big)/k$ in order to study a sharing problem is not new. Indeed, in \cite{Segal-Halevi_families}, E.~Segal-Halevi and S.~Nitzan use this average to define whether a sharing is fair for a group. An allocation is proportional on average when \mbox{$\big(\sum_{j \in J} \mu_i(X_j)\big)/k \geq t_J$} where $t_J$ is a positive number representing the entitlement of the group $J$.
Here, the problem under study is different; we are not studying equitable sharing for a group but for an agent. However, to define $k$-proportionality, we compare the value received by an agent to the average value received by the group to which he or she belongs.\\
In \cite{CHB}, the authors introduce two  hierarchies of fairness notions. The first one is Complement Harmonically Bounded (CHB) and the second Complement Linearly Bounded (CLB).
An allocation is CHB-$k$ if for any subset of agents $S$ of size at most $k$, and every agent $i \in  S$, the value given by the $i$-th agent for the union of
all pieces allocated to agents outside of $S$ is at most  $\frac{n-|S|}{n-|S|+1}$; for CLB-$k$ allocations, the upper bound becomes $\frac{n-|S|}{n}$.
The idea behind these definitions is that each agent believes that each group to which he or she does not belong receives a limited part of the cake. In this hierarchy, the order is not the same as the one introduced with $k$-proportionality. 
Indeed, CHB-$1$ and CLB-$1$ coincides with proportional division,  and CHB-$k$ (resp. CLB-$k$) is stronger than CHB-$(k-1)$ (resp. CLB-$(k-1)$). However, envy-free is stronger than CHB-$n$ and CLB-$n$.\\
$k$-proportionality, CHB-k and CLB-k are different scales.
With $k$-proportionality we compare the value obtained by an agent to the average value obtained by $k$ agents. With CHB-k or CHL-k, we compare the average value obtained by a group of agents to a constant depending on the size of the group.\\
In Sections~\ref{sec:equitable} and \ref{sec:pareto}, we will use $k$-proportionality to extend the results of \cite{Barbanel_Brams_Stromquist}. In \cite{BJK}, S.~Brams, M.~Jones, and C.~Klamler proved another result similar to those given in \cite{Barbanel_Brams_Stromquist}. In \cite{BJK}, the authors proved that there exist measures for which it is impossible to obtain an equitable, envy-free, and Pareto-optimal partition of a cake. In this result, $X_i$ is not assumed to be connected.
Here, the impossibility results will only concern cases where $X_i$ is connected. In Section 2, we generalize the impossibility of having a connected, envy-free, and equitable division of a pie. In Section 3, we generalize the impossibility of having an envy-free and Pareto-optimal partition of a cake. We therefore do not consider the general case where $X_i$ can consist of a finite number of intervals, but our theorems do not assume that the three properties: envy-free, equitable, and Pareto-optimal are all satisfied.

\section{First properties}

The following lemma shows that we can also consider $k$-proportionality as a measure of envy-freeness.

\begin{Lem}\label{lem:J'}
A division $X=X_1 \sqcup \ldots \sqcup X_n$ is $k$-proportional if and only if
for all $i$ and all subsets $J' \subset \{1, \ldots, n\} \setminus \{i\}$ such that $|J'|=k-1$, we have
$$\mu_i(X_i) \geq \dfrac{\sum_{j \in J'}\mu_i(X_j)}{k-1}.$$
\end{Lem}
This lemma says that in a $k$-proportional division a player does not envy the average value obtained by $k-1$ players.

\begin{proof}
With the notations introduced previously we have:
\begin{eqnarray*}
\mu_i(X_i) \geq \dfrac{ \sum_{j \in J} \mu_i(X_j)}{k} & \iff & \mu_i(X_i) \geq \dfrac{ \sum_{j \in J'} \mu_i(X_j)}{k}+\dfrac{\mu_i(X_i)}{k}\\
& \iff & \mu_i(X_i)\Big( 1 -\dfrac{1}{k}\Big) \geq \dfrac{ \sum_{j \in J'} \mu_i(X_j)}{k}\\
& \iff & \mu_i(X_i)  \geq \dfrac{ \sum_{j \in J'} \mu_i(X_j)}{k-1} \quad \quad \qedhere 
\end{eqnarray*}
\end{proof}

We have recalled that proportionality implies envy-freeness, the following lemma extends this property.

\begin{Lem}\label{lem:k-k+1}
If a division is $k$-proportional then it is also $k+1$-proportional.
\end{Lem}

\begin{proof}
Let $X=X_1 \sqcup \ldots \sqcup X_n$ be a $k$-proportional division.\\
Let  $K=\{i,j_1,\ldots,j_{k-1},j_k\}$, where $j_l \neq j_m$ if $l\neq m$ and $j_l \neq i$.\\
We have \mbox{$|K|=k+1$.}\\
We want to show that 
$$(\star) \quad (k+1) \mu_i(X_i) \geq \sum_{j \in K} \mu_i(X_j).$$
We set $J_l=K \setminus \{j_l\}$, where $l=1, \ldots, k$. 
As the division is $k$-proportional, we have for all $l=1, \ldots, k$
$$k \mu_i(X_i) \geq \sum_{j \in J_l} \mu_i(X_j).$$
When we sum all these inequalities we get
$$
k^2 \mu_i(X_i) \geq \sum_{l=1}^{k} \sum_{j \in J_l} \mu_i(X_j) =
k \mu_i(X_i)+ (k-1) \mu_i(X_{j_1}) + \cdots + (k-1) \mu_i(X_{j_k}).
$$
Thus 
$$(k^2-1)\mu_i(X_i) \geq (k-1) \mu_i(X_i) + (k-1)\mu_i(X_{j_1}) + \cdots + (k-1)\mu_i(X_{j_k}),$$
 and this implies the inequality $(\star)$.
\end{proof}

In the following, we are going to  represent the result of a division with a \emph{sharing matrix}.

\begin{Def}
A matrix $M = \big(\mu_i(X_j)\big)$ is said to be a sharing matrix when \mbox{$X=\sqcup_{j=1}^n X_j$} is a partition of $X$.
\end{Def}

\noindent\textbf{Example:}
Let $X=\sqcup_{i=1}^4 X_i$  be a division of $X$ between four players. If the sharing matrix associated to this partition is given by
$$M=\begin{pmatrix}
1/3 & 0 &0 & 2/3\\
0& 1/3 &0 &2/3\\
0&0& 1/3 & 2/3\\
1/3 &0&0& 2/3
\end{pmatrix}
$$
then this means that $\mu_1(X_1)=1/3$, $\mu_1(X_2)=\mu_1(X_3)=0$ and $\mu_1(X_4)=2/3$.\\
We can remark that this division is 4-proportional, ie proportional, it is also 3-proportional, but it is not 2-proportional. Indeed, the first three players envie the last player.

\section{$k$-proportionality, equitability and connected pieces}\label{sec:equitable}
In \cite{Cechexistence}, K.~Cechl{\'a}rov{\'a}, J.~Dobo{\v s}, and E.~Pill{\'a}rov{\'a} have shown that for all measures $\mu_1$, \ldots, $\mu_n$ there always exist a proportional and equitable division of $[0,1]$ with connected pieces $X_i$. We remark that this result is still valid when we consider a pie, i.e. $X=S^1$. Indeed, with one cut we restrict the study of $X=S^1$ to the case where $X=[0,1]$.\\
However, when  we consider envy-free division of $S^1$ we have the opposite result. Indeed, J. Barbanel, S. Brams and W. Stromquist have shown in \cite[Theorem 3.2]{Barbanel_Brams_Stromquist} that there exists measures $\mu_1$, \ldots, $\mu_n$, for which no division of the pie $S^1$ with connected components, is envy-free and equitable. In \cite{BJK2006}, S.~Brams, M.~Jones and C.~Klamler have shown that  there exists measures $\mu_1$, $\mu_2$, $\mu_3$, for which no division of the cake $[0,1]$ with connected components, is envy-free and equitable. \\
Thus, there is an obstruction when we go from the proportional division to the envy-free one. The question is then:  Where is this obstruction when we consider $k$-proportional division ?
 
\begin{Thm}
When $n \geq 5$, there exists measures $\mu_1$, \ldots $\mu_n$ for which no division of the pie $X=S^1=X_1 \sqcup \ldots \sqcup X_n$ with connected components $X_i$, is $(n-1)$-proportional and equitable.
\end{Thm}

 When we start with a proportional division and want to get closer to an envy-free division, the first step is the $(n-1)$-proportional division, the second the \mbox{$(n-2)$-proportional} division, and so on. Thus, this theorem says that the impossibility result appears in the first step.

\begin{proof}
This proof is a generalisation of the one given in \cite{Barbanel_Brams_Stromquist}.\\
In this proof, we consider $S^1$ as the interval $[0,1]$ where $0$ and $1$ are identified.\\
We consider the measures $\mu_1$, \ldots, $\mu_n$ on $S^1$ defined with the following corresponding densities:
$f_1(x)=\dfrac{6}{5} \mathcal{X}_{S^1\setminus [0,1/6]}$, $f_2(x)= \dfrac{6}{5} \mathcal{X}_{S^1\setminus [1/2,2/3]}$ and $f_i(x)=1$, for \mbox{$i=3, \ldots, n$}.\\
We suppose that a $(n-1)$-proportional, equitable and connected division \linebreak \mbox{$X=X_1 \sqcup \ldots \sqcup X_n$} exists and we are looking for a contradiction.\\

First, we claim that for $i=3, \ldots, n$, we have: 
$$X_i \cap [0, 1/6] \neq \emptyset \quad \textrm{ and } X_i \cap [1/2, 2/3] \neq \emptyset.$$
Indeed, for $i=3, \ldots, n$,  if $X_i \subset [1/6,1]$ then $\mu_1(X_i) >\mu_i(X_i)$. As \linebreak $\mu_i(X_i)=\mu_1(X_1)$, since the division is equitable, we get $\mu_1(X_i) >\mu_1(X_1)$.\\
Thus $\mu_1(X_3)+\cdots +\mu_1(X_n) > (n-2)\mu_1(X_1)$, but, by Lemma~\ref{lem:J'}, this is impossible since the division is $(n-1)$-proportional. Therefore $X_i \cap [0, 1/6] \neq \emptyset$\\
For the same reason we have $X_i \cap [1/2, 2/3] \neq \emptyset$.\\

Now,  as $X_i \cap [0, 1/6] \neq \emptyset$ and $X_i \cap [1/2, 2/3] \neq \emptyset$ we deduce that \linebreak $lg(X_i)\geq 1/3$, where $lg(X_i)$ is the length of $X_i$.\\
Thus $lg(X_3 \sqcup \ldots \sqcup X_n) \geq (n-2)/3$, because $X_i \cap X_j = \emptyset$ when $i \neq j$.\\
At last, as $n \geq 5$, the previous inequality gives $lg(X_3 \sqcup \ldots \sqcup X_n) \geq 1$.\\
If $n>5$, then inequality is strict and we have obtained the desired contradiction.\\
If $n=5$, then $lg(X_3 \sqcup \ldots \sqcup X_n) = 1$ and  this implies  $\mu_1(X_1)=0$. As the division is $(n-1)$-proportional, thus by Lemma~\ref{lem:k-k+1} $n$-proportional, we get $\mu_1(X_1) \geq 1/n$ and we get the desired contradiction.
\end{proof}

\section{$k$-proportionality, Pareto optimality and connected pieces}\label{sec:pareto}

S.~Banach and B.~Knaster have shown that for all  measures $\mu_1$, \ldots, $\mu_n$, there always exist a proportional division of $X$ with connected pieces $X_i$, see \cite{DubinsSpanier}. Therefore, the set of all partitions giving a connected proportional division  of the cake $X=[0,1]$ is  non-empty. Moreover, this set is compact. Thus, a connected proportional division for which the sum $\sum_i \mu_i(X_i)$ is maximum exists.\\
This partition then yields a Pareto-optimal, connected proportional division.
However, it is not always possible to obtain a Pareto-optimal, connected envy-free division, see \cite[Theorem 2.2]{Barbanel_Brams_Stromquist}.\\
Once again, we can ask where on our scale this impossibility first appears. The following theorem shows that, here again, the impossibility appears as soon as we have a $(n-1)$-proportional division.

\begin{Thm}
There exist measures for which no partition of the cake $X=[0,1]=\sqcup_{i=1}^n X_i$, with connected components $X_i$,  is $(n-1)$-proportional and Pareto optimal.
\end{Thm}

\begin{proof}
We consider $n$ players with associated measures $\mu_1$, $\mu_2$, \ldots, $\mu_n$.\\
The density on $[0,1]$ of $\mu_1$ is given by 
$$2 \times \mathcal{X}_{[0,1/2n]}(x)+\mathcal{X}_{[1/n,1]}(x).$$
The measures $\mu_2$, \ldots, $\mu_n$ are all equal to the Lebesgue measure  on $[0,1]$. Thus $\mu_i([a,b])=lg([a,b])=b-a$, for $i=2, \ldots,n$, where $lg([a,b])$ is the length of $[a,b]$.\\

A partition of the cake with connected components has the following form:\\ $X=\sqcup_{j=0}^{n-1} [x_j,x_{j+1}]$, where $x_0=0$, $x_{i+1} \geq x_i$ and $x_n=1$.\\
In the following, we are going to prove for these measures that: if we have a $(n-1)$-proportional division with connected pieces, then the division is not Pareto optimal.\\

First, we suppose that the first player gets the leftmost part of the cake.\\
 Thus $X_1=[0,x_1]$ and $[x_1,1]$ is divided between the last $n-1$ players. \\
 In this situation, for $i=2, \ldots, n-1$, we have:
$$(\star) \quad lg(X_i)=\dfrac{1-x_1}{n-1}.$$
Indeed, as the division is $(n-1)$-proportional, we have for $i=2, \ldots,n$: 
$$lg(X_i)=\mu_i(X_i) \geq \dfrac{\mu_i(X_2)+\cdots+\mu_i(X_n)}{n-1}=\dfrac{\mu_i(X\setminus X_1)}{n-1}=\dfrac{1-x_1}{n-1}.$$
Then, we sum all these inequalities  and if one of them is strict then we get
$$\sum_{i=2}^n lg(X_i)=lg([x_1,1])=1-x_1> 1-x_1.$$
As the previous inequality is impossible, this proves  $(\star)$.\\

Now, we prove that $x_1 = 1/n$.\\
As the division is $(n-1)$-proportional then, by Lemma~\ref{lem:k-k+1}, this division is also $n$-proportional and we get $lg(X_i)=\mu_i(X_i)\geq 1/n$, for $i=2, \ldots,n$. 
 Thus 
 $$1-x_1=lg([x_1,1])=lg(X_2)+\cdots+lg(X_n)\geq (n-1)/n.$$
This implies that  $x_1\leq 1/n$. Therefore, we just need to prove that $x_1<1/n$ is impossible.\\
Then, suppose that $x_1<1/n$.  \\
Thanks to $(\star)$, we get $x_2=x_1+\dfrac{1-x_1}{n-1}$. As $x_1<1/n$, we have  $\dfrac{1-x_1}{n-1} >\dfrac{1}{n}$. Thus $x_2>1/n$. Therefore $x_2, x_3, \ldots, x_n \in [1/n,1]$. By definition of $\mu_1$ and thanks to $(\star)$, we get $\mu_1(X_3)=\cdots=\mu_1(X_n)=\dfrac{1-x_1}{n-1}$.\\
Now, by Lemma~\ref{lem:J'}, the $(n-1)$-proportionality gives
$$\mu_1(X_1) \geq \dfrac{\mu_1(X_3)+\cdots+\mu_1(X_n)}{n-2}=\dfrac{1-x_1}{n-1}>\dfrac{1}{n}.$$
This gives a contradiction, because thanks to the definition of $\mu_1$, we have \mbox{$\mu_1(X_1)\leq 1/n$} since $x_1 <1/n$.\\

Therefore, if we have a $(n-1)$-proportional division and  the leftmost part is given to the first player then $x_1=1/n$. Thus, by definition of $\mu_i$ and thanks to $(\star)$, we get 
$$\mu_i(X_i)=1/n, \quad  \textrm{for } i=1, \ldots, n.$$
 However, this allocation is dominated by the division where the first player get $[0,1/2n]$ and $[1/2n,1]$ is divided between the other players into $n-1$ parts of the same length. More precisely,  we consider the following division $X=\sqcup_{i=1}^n \tilde{X}_i$ where 
$$\tilde{X}_1=\Big[0,\dfrac{1}{2n}\Big]$$
$$ \textrm{ and }\tilde{X}_i=\Big[ \dfrac{1}{2n} +\dfrac{i-2}{n-1}\Big( 1-\dfrac{1}{2n}\Big),  \dfrac{1}{2n} +\dfrac{i-1}{n-1}\Big( 1-\dfrac{1}{2n}\Big) \Big], \textrm{ for }i=2, \ldots,n.$$\\
With this division  we have $\mu_1(\tilde{X}_1)=1/n$ but $\mu_i(\tilde{X}_i)= \dfrac{1}{n-1}\Big(1-\dfrac{1}{2n}\Big)>\dfrac{1}{n}$, for $i=2, \ldots, n$.\\
Thus, if we want to get a Pareto optimal division, then the leftmost part cannot be given to the first player.\\

Now, we suppose that the leftmost part of the cake is not given to the first player.  For example, the leftmost piece $[0,x_1]$ is given to the last player.
As the division is \mbox{$(n-1)$-proportional}, it is also $n$-proportional and we get: 
$$\mu_n(X_n)=lg([0,x_1]) \geq 1/n.$$
 Thus $x_1\geq 1/n$.\\
 Furthermore, on $[x_1,1]$ all the measures coincides with the Lebesgue mesure. Therefore,  using $(n-1)$-proportionality as before, we can show that $lg(X_i)=\dfrac{1-x_1}{n-1}$ for $i=1,\ldots, n-1$.
As the division is $n$-proportional, we get 
$$\dfrac{1-x_1}{n-1} \geq \dfrac{1}{n}\Rightarrow \dfrac{1}{n} \geq x_1.$$
Thus, $x_1=1/n$ and $\mu_i(X_i)=1/n$ for $i=1, \ldots,n$.\\
As before, this allocation is dominated by the division $X=\sqcup_{i=1}^n \tilde{X}_i$ given previously.
\end{proof}

\section{Strong $k$-proportionality}\label{sec:strong}
In this section, we are going to prove that with  $k$-proportionality we can give  a uniform statement for well known results about strong proportional and strong envy-free divisions.

First, we define the natural generalization of strong proportional and strong envy-free division.

\begin{Def}
We say that a division is \emph{strong $k$-proportional} if for all \mbox{$J \subset \{1, \ldots,n\}$} such that $|J|=k$ and all $i \in J$, we have 
$$\mu_i(X_i) > \dfrac{\sum_{j \in J} \mu_i(X_j)}{k}.$$
\end{Def}

Now we are going to prove the following:

\begin{Thm}\label{thm:strong_k-prop}
When consider $n$ measures $\mu_1$, \ldots, $\mu_n$, we have the following equivalence: there exists a strong $k$-proportional division of the cake for the measures \mbox{$\mu_1$,\ldots, $\mu_n$} if and only if
all subsets $S \subset \{\mu_1, \ldots, \mu_n\}$, such that $|S|=k$, contains two different measures.
 \end{Thm}

When $k=n$, this theorem gives: There exists a strong proportional division if and only if two measures are differents. This gives the result given by L.~Dubins and E.~Spannier in \cite[Corollary 1.2]{DubinsSpanier}. An algorithm that provides, when it exists, a strong proportional division with connected parts has been given in \cite{Strong_prop}.\\
When $k=2$, this theorem gives: There exists a strong envy-free division if and only if  the measures are all different. This corresponds to the result given by J.~Barbanel, see \cite[Theorem 5.5]{Barbanel}.\\
We can sum up Theorem~\ref{thm:strong_k-prop} by saying that the more players there are with different tastes, the easier it is to avoid jealousy.\\

\textbf{Exemple}: Consider 6 measures $\mu_1$, \ldots, $\mu_6$ such that $\mu_1$, $\mu_2$ and $\mu_3$ are different, ie $| \{ \mu_1, \mu_2, \mu_3 \} |=3$. Moreover, suppose that $\mu_1=\mu_4$, $\mu_2=\mu_5$, $\mu_3=\mu_6$.
Theorem~\ref{thm:strong_k-prop} implies that there exists a strong 3-proportional division for these measures but there do not exist a strong-envy free divisions.\\

A part of the proof of this theorem relies on a theorem due to Barbanel. Barbanel's theorem uses the notion of \emph{proper matrices}.

\begin{Def}
Suppose that $Q=(q_{ij})$ is an $n \times n$ matrix of real numbers.\\
We say that $Q$ is a \emph{proper matrix} relatively to $\mu_1$, \ldots, $\mu_n$ if:
\begin{itemize}
\item each row of $Q$ sums to zero,
\item each column of $Q$ is consistent with the linear equations involving \mbox{$\mu_1$, \ldots, $\mu_n$}. That is to say, if $\sum_i \lambda_i \mu_i=0$ then $\sum_i \lambda_i q_{ij}=0$.
\end{itemize}
\end{Def}

Now we recall Barbanel's theorem, see~\cite[Theorem 4.18]{Barbanel}.
\begin{Thm}[Barbanel]\label{thm:barbanel}
If $Q=(q_{ij})$ is a proper matrix relatively to $\mu_1$, \ldots, $\mu_n$, then for some $\varepsilon>0$ the matrix with coefficients $(\frac{1}{n} + \varepsilon q_{ij})$ is a sharing matrix.
\end{Thm}

We denote by  $E$  the matrix with all coefficients equal to $1/n$. This matrix is the sharing matrix associated to an exact division, ie $\mu_i(X_j)=1/n$.  Therefore, Barbanel's theorem says that for some $\varepsilon>0$ the matrix $E+\varepsilon Q$ is  a sharing matrix. Thus there exists a partition $X=\sqcup_{j=1}^n X_j$ such that we have the matrix equality $\big( \mu_i(X_j)\big)=E+\varepsilon Q$. Then $E +\varepsilon Q$ is just the sharing matrix associated to a small perturbation of an exact division.\\

Furthermore, we also have the following lemma, see~\cite[Lemma 5.6]{Barbanel}.

\begin{Lem}\label{lem:barbanel}
There exists a proper matrix $Q=(q_{ij})$ such that for all $i,j=1, \ldots,n$:
$$q_{ii} \geq q_{ij} \textrm{ with equality holding if and only if } \mu_i=\mu_j.$$ 
\end{Lem}

\begin{proof}[Proof of Theorem~\ref{thm:strong_k-prop}]$\,$\\
$(\Rightarrow)$ We suppose that there exists a strong $k$-proportional division and that there exists a subset $S$ such that $|S|=k$ and all measures in $S$ are equal. We are going to prove that this situation is impossible.\\
In order to simplify the notations, we set $S=\{\mu_1, \ldots, \mu_k\}$. The strong $k$-proportionality gives 
$$\mu_{i}(X_{i}) >\dfrac{\mu_{i}(X_{1})+\cdots+\mu_{i}(X_k)}{k}= \dfrac{\mu_{i}(\sqcup_{j=1}^l X_{j})}{k}, \textrm{ for } i=1, \ldots, k.$$
 The sum of all these inequalities and the equalities $\mu_{1}=\mu_{2}=\cdots=\mu_{k}$  give
$$\mu_{1}(\sqcup_{j=1}^k X_{j} ) > \mu_{1}(\sqcup_{j=1}^k X_{j} ).$$
This gives the desired contradiction.\\
$(\Leftarrow)$ 
Thanks to Theorem~\ref{thm:barbanel} and Lemma~\ref{lem:barbanel}, we can deduce that there exists a partition such that $\mu_i(X_i) \geq \mu_i(X_j)$ and $\mu_i(X_i)> \mu_i(X_j)$ if and only if $\mu_i \neq \mu_j$.\\
Consider a set $J$ such that $i \in J$ and $|J|=k$, then, by hypothesis, in this set there exist a measure $\mu_{j_0}$ such that $\mu_i \neq \mu_{j_0}$.  Thus $\mu_i(X_i) > \mu_i(X_{j_0})$. Therefore, in the sum $\sum_{j \in J} \mu_i(X_j)$, we have $\mu_i(X_i) \geq \mu_i(X_j)$ for all $j \in J$ and  there exist at least one term $\mu_i(X_{j_0})$ such that $\mu_i(X_i) > \mu_i(X_{j_0})$. \\
It follows that $\mu_i(X_i) > \frac{1}{k}\sum_{j \in J} \mu_i(X_j)$, thus the partition is strong $k$-proportional.
\end{proof}

\section{Conclusion}
In this article, we have introduced a scale that allows us to consider intermediate fair divisions between proportional division ($k=n$) and envy-free division $(k=2)$.\\
This has enabled us, for example, to give a general formulation for the characterization of strong proportional and strong envy-free division.
It has also enabled us to see that certain impossibility results appear when $k=n-1$. So, these impossibility results appear as soon as we no longer consider a proportional division but a division that is a little more demanding. More precisely, impossibility results arise when we simply want each group of $n-1$ players to consider the distribution to be proportional. \\

In the definition of proportional division, for each player we have an inequality that is independent of the measure of the other players. Indeed, we have $\mu_i(X_i)\geq 1/n$.
When we consider a $k$-proportional division with $k\leq n-1$, each player compares what he or she has received with a quantity that depends on what the others have received.
So what our theorems show, is that it's impossible to have certain divisions as soon as players take into account what the others have received. It is therefore not necessary to compare the value of one's share with that of the others (as in the case of envy-free division) to have a situation that may be impossible to resolve. As soon as we take into account the value of the other players, sharing may become impossible.\\
A moral would therefore be as follows: to avoid difficulties when sharing, each player must be unaware of what the others have received\ldots\\

The study we have begun here can be extended in many ways.
We can ask ourselves whether the problems arise systematically for $(n-1)$-proportional sharing or whether sometimes the difficulty comes later.
For example, there are algorithms for constructing a proportional and connected fair division, but we know that there are no algorithm for constructing a connected and envy-free division, see \cite{Stromquist}. Again, does the impossibility arise as soon as we consider a $(n-1)$-proportional division, or does it occur for some other value on our scale?\\
On the other hand, when we do not want $X_i$ to be connected, there exist envy-free division algorithms. One way of measuring the difficulty of envy-free division versus proportional division is to compare the complexities of the algorithms producing these divisions. Complexity counts the number of questions asked to the players. Even-Paz algorithm, see \cite{EvenPaz}, produces proportional divisions with a number of queries in $O(n\log(n))$. For envy-free division, Aziz-Mackenzie's algorithm  has a complexity in $O(n^{n^{n^{n^{n^n}}}})$, see \cite{AzizMackenzie}.\\
 The difference between these two complexities is huge and this leads to the folowing questions: What can we say about $k$-proportional sharing? \\

In conclusion, we hope that the new forms of sharing introduced between proportional division and envy-free division will enable us to better understand the gap between statements about proportional division and envy-free division.\\

\textbf{Conflict of interest:}
The author declare that he has no conflict of interest
related to this manuscript.\\

\textbf{Funding informations:} No funding was received for this article.


\bibliographystyle{plain}
\bibliography{cakebiblio}

\end{document}